\documentclass[review]{elsarticle}
\usepackage{hyperref} 
\usepackage{mathrsfs}
\usepackage{titlesec}
\usepackage[thinlines,thiklines]{easybmat}
\usepackage{graphicx}
\usepackage{float}
\usepackage{color}
\usepackage{braket}
\usepackage{xcolor}
\usepackage{tablefootnote}
\usepackage{graphicx}
\usepackage{amsmath}
\usepackage{amssymb}
\usepackage{amsthm}
\usepackage{amsfonts}
\usepackage{tcolorbox}
\usepackage{colortbl}
\theoremstyle{plain}\newtheorem{theorem}{Theorem}
\newtheorem{Lemma}{Lemma}
\newtheorem{Claim}{Claim}
\newtheorem{Definition}{Definition}
\newtheorem{Fact}{Fact}
\newtheorem{Remark}{Remark}
\newtheorem{Example}{Example}
\usepackage{algorithm}
\usepackage{algorithmic}

\begin{document}
\begin{frontmatter}
\title{Query complexity of generalized Simon's problem}
\author[a1]{Zekun Ye}
\author[a1]{Yunqi Huang}
\author[a1,a2]{Lvzhou Li\corref{one}}
\author[a3,a4]{Yuyi Wang}

\cortext[one]{Corresponding author. Mobile: +8613802437672; Corresponding address: School of Computer Science and Engineering, Sun Yat-sen University, Guangzhou 510006,
 China\\
 \indent {\it E-mail
address:} yezekun@mail2.sysu.edu.cn (Z. Ye); lilvzh@mail.sysu.edu.cn (L.
Li)}

\address[a1]{Institute of Quantum Computing and Computer Science Theory, School of Computer
Science and Engineering, Sun Yat-sen University, Guangzhou 510006, China}
\address[a2]{Ministry of Education Key Laboratory of Machine Intelligence and Advanced Computing (Sun Yat-sen University), Guangzhou 510006, China}
\address[a3]{Disco Group, ETH Zurich, Switzerland}
\address[a4]{CRRC Zhuzhou Institute}

\begin{abstract}
Simon's problem plays an important role in the history of quantum algorithms, as it inspired Shor to discover the celebrated quantum algorithm solving integer factorization in polynomial time. Besides, the quantum algorithm for Simon's problem has been recently applied to break symmetric cryptosystems. Generalized Simon's problem, denoted by $\mathsf{GSP}(p,n,k)$, is a natural extension of Simon's problem: Given a function $f:\mathbb{Z}_p^n \to X$ where $X$ is a finite set and the promise that for any $x, y \in \mathbb{Z}_p^n, f(x) = f(y)$ iff $x - y \in S$ for a subgroup $S \le \mathbb{Z}_p^n$ of rank $k<n$, the goal is to find $S$. In this paper we consider the query complexity of the problem, that is, the minimum number of queries to $f$ required to find $S$. First, it is not difficult to design a quantum algorithm solving the above problem with query complexity of $O(n-k)$. However, so far it is not clear what is the classical query complexity of the problem, and revealing this complexity is necessary for clarifying the computational power gap between quantum and classical computing on the problem.

To tackle this problem, we prove that any classical (deterministic or randomized) algorithm for $\mathsf{GSP}(p,n,k)$ has to query at least $\Omega\left(\max\{k, \sqrt{p^{n-k}}\}\right)$ values and any classical nonadaptive deterministic algorithm for $\mathsf{GSP}(p,n,k)$ has to query at least $\Omega\left(\max\{k, \sqrt{k \cdot p^{n-k}}\}\right)$ values. Hence, we clearly show the classical computing model is less powerful than the quantum counterpart, in terms of query complexity for the generalized Simon's problem. Moreover, we obtain an upper bound $O\left(\max\{k, \sqrt{k \cdot p^{n-k}}\}\right)$ on the classical deterministic query complexity of $\mathsf{GSP}(p,n,k)$, by devising a subtle classical algorithm based on group theory and the divide-and-conquer approach. Therefore, we have an almost full characterization of the classical deterministic query complexity of the generalized Simon’s problem.

\end{abstract}


\end{frontmatter}

	
\section{Introduction}
Query complexity, also called decision tree complexity \cite{buhrman2002complexity}, is the computational complexity of a problem or algorithm expressed in terms of the decision tree model. It has been very useful for understanding the power of different computational models. In contrast to the Turing machine world where lower bounds and separations between complexity classes often have to rely on unproven conjectures, using query complexity can often prove tight lower bounds and have provable separations between different computational models. For instance, quantum computing has been shown to have exponential advantages over classical computing in terms of query complexity. 

Simon's problem is well known in the history of quantum algorithms. It shows an exponential gap between classical and quantum computing. The definition of Simon's problem is as follows. 
\begin{tcolorbox}
			\label{Simon's problem}
			\noindent
			\textbf{Given:} An (unknown) function $f:\{0,1\}^n\to \{0,1\}^m$.
			\\
			\\
			\textbf{Promise:} There exists a nonzero element $s\in \{0,1\}^n$ such that for all $g,h\in\{0,1\}^n, f(g) = f(h)$ iff $g = h$ or $g = h \oplus s$.
			\\
			\\
			\textbf{Problem:} Find $s$.
		\end{tcolorbox}
Simon \cite{simon1994power} proposed a quantum algorithm to efficiently find $s$. Shortly afterward, inspired by this algorithm, Shor discovered the celebrated quantum algorithm for the integer factorization problem \cite{shor1994algorithms}. Actually, Simon's problem and its variants have attracted a lot of attention from the academic community \cite{alagic2007quantum,brassard1997exact,cai2018optimal,kaye2006introduction,koiran2005quantum,mihara2003deterministic,santoli2016using,wu2019quantum}. More details will be introduced in the section of related work soon.

\subsection{Problem statement and our results}\label{Problem}
Several extended versions of Simon's problem with a minor difference 	have been studied from the viewpoint of quantum computing \cite{alagic2007quantum,brassard1997exact,kaye2006introduction,mihara2003deterministic,wu2019quantum}. An extended version of Simon's problem can be described as: Given an (unknown) function $f:\mathbb{Z}_2^n \to X$ with  a finite set $X$ and a positive integer $k <n$, it is promised that there exists a subgroup $S \le \mathbb{Z}_2^n$ of rank $k$ such that for any $x, y \in \mathbb{Z}_2^n, f(x) = f(y)$ iff $x \oplus y \in S$, and then the goal is to find $S$. This problem is a natural extension since Simon's problem is a special case of it with $k=1$.

\begin{Example}
		We present an example of the above generalized Simon's problem in Table \ref{table:ex}. In this example, $n = 4$, $k = 2$, and we can find that $S = \{0000, 0011, 0110, 0101\}$.
		
		\begin{table}[H]
			\caption{An example of the generalization of Simon's problem}
			\label{table:ex}
			\centering
			\begin{tabular}{cccc|c}
				\hline
				\multicolumn {4}{c|}{$x$} &$f(x)$ \\
				\hline
				0000 & 0011& 0110 & 0101 & 0000\\
				0001 & 0010 & 0111& 0100 & 0001\\
				1000 & 1011 &1110 & 1101 & 0010\\
				1001 & 1010 & 1111 & 1100 & 0011 \\
				\hline
			\end{tabular}
		\end{table}
		
	\end{Example}

In this paper, we consider a more general extension of Simon's problem, called a \textit{generalized Simon's problem} as follows. In this case, we generalize the domain of $f$ from $\mathbb{Z}_2^n$ to $\mathbb{Z}_p^n$, where $p$ is a prime\footnote{{Our method relies on the fact that $\mathbb{Z}_p^n$ is a vector space over the finite field $\mathbb{Z}_p$. Note that $\mathbb{Z}_p$ is a finite field if and only if $p$ is a prime.}}.
\begin{tcolorbox}
			\label{generalized Simon's problem}
			\noindent
			\textbf{Given:} An (unknown) function $f:\mathbb{Z}_p^n \to X$, where $X$ is a finite set, and a positive integer $k <n$.
			\\
			\\
			\textbf{Promise:} There exists a subgroup $S \le \mathbb{Z}_p^n$ of rank $k$ such that for any $x, y \in \mathbb{Z}_p^n, f(x) = f(y)$ iff $x - y \in S$.
			\\
			\\
			\textbf{Problem:} Find $S$.
		\end{tcolorbox}
	
	For convenience, we use $\mathsf{GSP}(p,n,k)$ to denote the above problem with parameters $n, p, k$ throughout this paper. 
	It is easy to see that Simon's problem is a special case of $\mathsf{GSP}(p, n, k)$ with $p = 2$ and $k = 1$.

	To find the subgroup $S$, we need to design an algorithm that is allowed to access the function $f$ by querying an oracle that, given $x$, outputs $f(x)$. According to randomness, algorithms can be divided into deterministic algorithms and randomized algorithms. Deterministic algorithms solve the problem with certainty, whereas randomized algorithms solve the problem with bounded error probability. Moreover, we also study a class of widely studied algorithms, nonadaptive algorithms, where current queries are not allowed to depend on the result of previous queries.
	
	The {(nonadaptive)} deterministic query complexity of $\mathsf{GSP}(p,n,k)$ is the query complexity of the optimal {(nonadaptive)} deterministic algorithm for that, and the query complexity of a {(nonadaptive)} deterministic algorithm is the number of queries it makes on the worst-case input. {Similarly, the randomized query complexities of $\mathsf{GSP}(p,n,k)$ is the query complexity of the optimal randomized algorithm for solving $\mathsf{GSP}(p,n,k)$ with bounded error, and the query complexity of a randomized algorithm is the maximum number of queries it makes, with the maximum taken over  both the choices of input and the internal randomness of the algorithm. A more detailed introduction about query complexity can be referred to Ref. \cite{buhrman2002complexity}.} 
	
	
  In this paper, we obtain some characterizations of the query complexity of $\mathsf{GSP}(p,n,k)$ in the following theorems. 
   

	\begin{theorem} [Lower bound]
		\label{Th:lower}
		Any classical (deterministic or randomized) algorithm solving $\mathsf{GSP}(p,n,k)$ needs to make $\Omega\left(\max\{k, \sqrt{p^{n-k}}\}\right)$ queries.
	\end{theorem}

	\begin{theorem} [Lower bound]
		\label{Th:lower_non}
		Any nonadaptive classical deterministic algorithm solving $\mathsf{GSP}(p,n,k)$ needs $\Omega\left(\max\{k,\sqrt{k\cdot p^{n-k}}\}\right)$ queries.
	\end{theorem}
	
	\begin{theorem}[Upper bound] 
		\label{Th:upper}
		There exists a classical deterministic algorithm to solve $\mathsf{GSP}(p,n,k)$ using $O\left(\max\{k,\sqrt{k\cdot p^{n-k}}\}\right)$ queries.
	\end{theorem}
	
	\begin{Remark}
		When $p=2$ and $k = 1$, the query complexity is $\Theta(\sqrt{2^{n}})$ which is the same as the result presented in \cite{cai2018optimal,Wolf2019quantum}. 
	\end{Remark}
	
	To prove the theorems above, we first use the double-counting method and the adversary method to obtain the lower bounds. Furthermore, we design a classical deterministic algorithm, obtaining the upper bound of the query complexity. Our algorithm is based on group theory and the divide-and-conquer technique. 

\begin{table}[H]
			\caption{Known results about the query complexity of Simon's problem and generalized Simon's problem. The randomized, deterministic and nonadaptive deterministic query complexities of $\mathsf{GSP}(p,n,k)$ are obtained in this paper.}
			\label{table: main results}
			\centering
			\scalebox{0.63}{
			\begin{tabular}{c|c|c|c}
				\hline
				\ & quantum & randomized/deterministic & nonadaptive deterministic\\
				\hline
				Simon's problem & $\Theta(n)$ \cite{brassard1997exact, cai2018optimal, koiran2005quantum, mihara2003deterministic, simon1994power} &
				\multicolumn {2}{c}{$\Theta(\sqrt{2^n})$\cite{cai2018optimal,Wolf2019quantum}} \\
				\hline
				$\mathsf{GSP}(p,n,k)$ & $\Theta(n-k)$\tablefootnote{The work in \cite{Hirvensalo2001quantum} implies that $\mathsf{GSP}(p,n,k)$ can be solved in $O(n-k)$ queries by a quantum algorithm with bounded error. Moreover, it can be solved by an exact quantum algorithm with $O(n-k)$ queries, if we slightly adjust the algorithms in \cite{brassard1997exact} (see \ref{exact quantum algorithm}). Additionally, by generalizing the lower bound method of \cite{koiran2005quantum}, it was found that any quantum algorithm for $\mathsf{GSP}(p,n,k)$ has to query at least $\Omega(n-k)$ times with $p=2$ \cite{wu2019quantum}. It is trivial to obtain the same result for general $p$.}
				& 
				$\Omega\left(\max\{k, \sqrt{p^{n-k}}\}\right)$, $O\left(\max\{k,\sqrt{k\cdot p^{n-k}}\}\right)$ & $\Omega\left(\max\{k, \sqrt{k\cdot p^{n-k}}\}\right)$\\
				\hline
			\end{tabular}
			}
		\end{table}

	\subsection{Motivation}
	Simon's problem plays a key role in the history of quantum algorithms and recently it has been found useful for applications in cryptography. First, the quantum algorithm \cite{simon1994power} for this problem inspired the discovery of Shor's algorithm \cite{shor1994algorithms} solving integer factorization in polynomial time, and also inspired Buhrman et al.\ \cite{Buhrman2008quantum} to propose the first exponential speed-up for quantum property testing. 
	Furthermore, Simon's algorithm can be used to show the insecurity of commonly used cryptographic symmetric-key primitives. For example, Kaplan et al. \cite{kaplan2016breaking} showed that several classical attacks based on finding collisions can be exponentially speeded up using Simon’s algorithm. 
	Moreover, it was used to break the 3-round Feistel construction \cite{Kuwakado2010quantum} and then to prove that the Even-Mansour construction \cite{Kuwakado2012security} is insecure with superposition queries. Also, it can be used to quantum related-key attacks \cite{roetteler2015note}.

	
	$\mathsf{GSP}(p,n,k)$ is a natural extended version of Simon's problem and also a special case of the hidden subgroup problem. The motivations for studying $\mathsf{GSP}(p,n,k)$ are as follows.
	
	\begin{itemize}
		\item First, from the viewpoint of quantum computing, characterizing the classical query complexity of $\mathsf{GSP}(p,n,k)$ is necessary for clarifying the computational power gap between quantum and classical computing on this problem. Note that we have a comprehensive understanding of the quantum query complexity for $\mathsf{GSP}(p,n,k)$ (see Table \ref{table: main results}). However, as far as we know, how well a classical algorithm performs on this problem still needs to be explored. 
		
		\item Second, from the viewpoint of classical computing, it is an interesting problem to reveal the query complexity of $\mathsf{GSP}(p,n,k)$, as it generalizes the well-studied Simon's problem and the ideas to solve the generalized problems are completely different from that for the original one. $\mathsf{GSP}(p,n,k)$ is a more general problem in mathematics, and the hidden law behind it is worth exploring. Based on our in-depth study, we find this problem seems not so straightforward to solve, instead, it involves several techniques, such as the double-counting method, the adversary method, and the divide-and-conquer method.
	\end{itemize}
	
	\subsection{Related work}
	Query complexity is a very successful measurement to study the relative power of quantum and classical computing. For quantum computing, query models can be divided into bounded-error and exact versions in terms of their outputs. A bounded-error model requires that the algorithm gives the correct result with some sufficiently high probability, while an exact model means that the algorithm gives the correct result with certainty. For the bounded-error case, there is much work showing the advantage of quantum algorithms over classical ones in terms of query complexity, e.g.,  \cite{Aaronson2018problem,aaronson2016separations,bun2018the}. However, the results for exact query algorithms seem more limited. For total Boolean functions, Beals et al. \cite{beals2001quantum} showed that quantum query algorithms only achieve polynomial speed-up over classical counterparts. In 2013, Ambainis \cite{ambainis2016superlinear} presented the first example that exact quantum computing has a superlinear advantage over deterministic algorithms. This result has been further improved by \cite{Ambainis2017separations} with a quadratic gap between its quantum and deterministic query complexity. On the other hand, for computing partial functions, there exists a super-exponential separation as shown by the famous Deutsch-Jozsa algorithm \cite{deutsch1992rapid}: constant versus linear. 
	
	In 1994, Simon's problem was shown to be solvable on quantum computers with $O(n)$ queries in the bounded-error setting. The lower bound of query complexity was proved to be $\Omega(n)$ in \cite{koiran2005quantum} by applying the polynomial method \cite{beals2001quantum}. For exact quantum query algorithms, Brassard and H\o yer \cite{brassard1997exact} solved the problem with $O(n)$ queries. Compared with their algorithm, Mihara and Sung \cite{mihara2003deterministic} proposed a simpler exact quantum algorithm. Apart from quantum algorithms, Cai and Qiu \cite{cai2018optimal} designed a classical deterministic algorithm for solving Simon's problem with $O(\sqrt{2^n})$ queries and proved their algorithm is optimal in terms of query complexity.
	
	Moreover, Simon's problem is a special case of a well-studied class of problems, the so-called \textit{hidden subgroup problem}. That is, given a function $f: G\to X$, where $G$ is a finitely
  generated group and $X$ is a finite set, such that $f$ is bijective on $G/S$ for a subgroup $S \le G$, the goal is to find $S$. Jozsa \cite{jozsa1998quantum} provided a uniform description of several important quantum algorithms such as Deutsch-Jozsa \cite{deutsch1992rapid}, Simon \cite{simon1994power}, and Shor \cite{shor1994algorithms} algorithms in terms of the hidden subgroup problem. Indeed, this problem has received a lot of attention, and many quantum algorithms were proposed for its different variants, e.g., \cite{bacon2005from,childs2007quantum,ettinger2004the,goncalves2017an,grigni2001quantum,hallgren2003the,kempe2005the,kuperberg2005a}.
	
	In Simon's problem, it is assumed that $G = \mathbb{Z}_2^n$ and $|S| = 2$. There are two directions to generalize Simon's problem. Alagic et al.\ \cite{alagic2007quantum} generalized Simon's problem by assuming $G = K^n$, where $K$ is a non-Abelian group of constant size, and $S$ is either trivial or $|S| = 2$. Additionally, they proposed an efficient quantum algorithm with time complexity $2^{O(\sqrt{n\log n})}$ to solve this extended version of Simon's problem. Brassard and H\o yer \cite{brassard1997exact} proposed another extended version of Simon's problem by assuming $G = \mathbb{Z}_2^n$ and $S \le G$, and solved this problem with $O(n)$ query complexity. Also, a similar result has been obtained in \cite{mihara2003deterministic}. Moreover, Brassard and H\o yer \cite{brassard1997exact} generalized their results to any finite general additive group. 
	
	In this paper, we consider the generalized Simon's problem similar to the problem in \cite{brassard1997exact}. It is assumed that $G = \mathbb{Z}_p^n$ and $rank(S) = k$, and the problem is denoted by $\mathsf{GSP}(p,n,k)$ as mentioned before. It is not difficult to see that $\mathsf{GSP}(p,n,k)$ can be solved exactly by adjusting the algorithms of \cite{brassard1997exact} slightly (see \ref{exact quantum algorithm}), and be solved with bounded error by a generalized Simon's algorithm \cite{Hirvensalo2001quantum}. All these algorithms need $O(n-k)$ queries. Meanwhile, by generalizing the method of \cite{koiran2005quantum}, Wu et al. \cite{wu2019quantum} found that the lower bound of quantum complexity of $\mathsf{GSP}(p,n,k)$ is $\Omega(n-k)$ with $p=2$. It is easy to obtain the same result for general $p$.
	
	To clarify the computational power gap between quantum and classical computing on this problem, there ought to be a characterization of the classical query complexity of it. However, to our knowledge, almost no related result has been obtained, except for the work \cite{wu2019quantum}. Compared with this work, our results are more general, since we consider the case $G = \mathbb{Z}_p^n$ rather than $G = \mathbb{Z}_2^n$. Meanwhile, in terms of classical query complexity, completely different techniques are used to obtain lower bounds, and we construct a smarter algorithm to obtain a better upper bound.
	Furthermore, if their lower bound proof was correct, then we could generalize their method to obtain the tight bound of the classical query complexity of $\mathsf{GSP}(p,n,k)$. However, it is unfortunate that we could not verify its correctness. 
	
	\subsection{Organization}
	The remainder of the paper is organized as follows. In Section \ref{sec:Preliminaries}, we review some notations concerning group theory used in this paper. In Section \ref{sec:lowerbound}, we present the lower bounds of deterministic, randomized and nonadaptive deterministic query complexity of $\mathsf{GSP}(p,n,k)$. In Section \ref{sec:upperbound}, an upper bound is obtained by giving a deterministic algorithm. Finally, a conclusion is made in Section \ref{sec:conclusion}. For completeness, an exact quantum query algorithm for $\mathsf{GSP}(p,n,k)$ is given in \ref{exact quantum algorithm}.
	
	\section{Preliminaries}
	\label{sec:Preliminaries}
	In this section, we present some notations used in this paper. Let $\mathbb{Z}_p$ denote the additive group of elements $\{0,1,...,p-1\}$ with addition modulo $p$, and $\mathbb{Z}_p^*$ denote the multiplication group of elements $\{1,2,...,p-1\}$ with multiplication modulo $p$. In the following, all the groups we mention are $\mathbb{Z}_p^n$ or its subgroup without special instructions, where $p$ is a prime. Let $x,y \in \mathbb{Z}_p^n$ with $x = (x_1,x_2,...,x_n)$ and $y = (y_1,y_2,...,y_n)$. For $x,y \in \mathbb{Z}_p^n$ with $x = (x_1,x_2,...,x_n)$ and $y = (y_1,y_2,...,y_n)$, we define
	$$
	x + y: = ((x_1 + y_1)\bmod p, (x_2 + y_2)\bmod p,...,(x_n + y_n)\bmod p),
	$$
	$$
	x - y:= x+(-y).
	$$
	For $X,Y \subseteq \mathbb{Z}_p^n,w\in\mathbb{Z}_p^n$, we define
	$$
	X \setminus Y := \{x\;|\; x\in X \land x\notin Y\},
	$$
	$$
	X + Y := \{x + y\;|\; x\in X, y\in Y\},
	$$
	$$
	X+w := \{x + w\;|\; x\in X\},
	$$
	$$
	X-w := X+(-w).
	$$
	We also use the abbreviated notation $X_{w}$ for $X+w$ and $X\cup w$ for $X \cup \{w\}$.
	Let $\alpha x$ denote $x+x+\cdots+x$ (the number of $x$ is $\alpha$).
	By $\langle X \rangle$, we denote the subgroup generated by $X$, i.e.,
	$$
	\langle X \rangle := \left\{ \sum_{i=1}^{k}\alpha_{i}x_i|x_i \in X,\alpha_{i} \in\mathbb{Z}_p \right\}.
	$$
	The set $X$ is called a generating set of $\langle X \rangle$. A set $X$ is \emph{linearly independent} if $\langle X \rangle \neq \langle Y \rangle$ for any proper subset $Y$ of $X$. In other words, a set $X$ is linearly independent if $X$ is the smallest generating set of $\langle X \rangle$. Notice that the cardinality $|\langle X \rangle |$ is $p^{|X|}$ if $X$ is linearly independent. 
	
	For any group $G$, the basis of $G$ is a maximal linearly independent subset of $G$.
	The cardinality of the basis of $G$ is called its \emph{rank}, denoted by $rank(G)$. If $H$ is a \emph{subgroup} of $G$, then we write $H \le G$; if $H$ is a \emph{proper subgroup}, then $H < G$. 
	\begin{Definition}[Complement Subgroup]
		\label{def:complement_group}
		For a group $G$ and its subgroup $H$, a group $\overline{H}_G$ is called a complement subgroup of $H$ in $G$ if
		$
		H + \overline{H}_G = G \; \text{and} \; H \cap \overline{H}_G = \{0^n\}.
		$
	\end{Definition}
	
	In this paper, we abbreviate $\overline{H}_G$ as $\overline{H}$ when $G = \mathbb{Z}_p^n$.

	To obtain our results, we need the following facts to characterize the query complexity of $\mathsf{GSP}$. 
	\begin{Fact}
		\label{Fact:exgroup}
		Let $V,W$ be two subgroups of $\mathbb{Z}_p^n$ such that $V \cap W = \{0^n\}$, $X$ be a basis of $V$ and $Y$ be a basis of $W$. Then $rank(V+W) = rank(V)+rank(W)$ and $X \cup Y$ is a basis of $V + W$.
	\end{Fact}
	
	\begin{Fact}
		\label{Fact:linear}
		Let $\{x_i\} \subseteq \mathbb{Z}_p^n$. If $\sum_{i}\alpha_{i}x_i = 0$ implies $\alpha_i = 0$ for any $i$, then $\{x_i\}$ is linearly independent. 
	\end{Fact}
	
	\begin{Fact}
		\label{Fact:join}
		Suppose $V,H$ {are} two subgroups of $\mathbb{Z}_p^n$ such that $V \cap H = \{0^n\}$. For $w \notin V$, we have $\langle V \cup w \rangle \cap H = \{0^n\}\Leftrightarrow V_w \cap H = \emptyset $.
	\end{Fact}
	
	\begin{proof}
		Suppose $\langle V \cup w \rangle \cap H = \{0^n\}$. Since $V_w \subseteq \langle V \cup w \rangle$, we have $V_w \cap H \subseteq \{0^n\}$. Because $0^n \notin V_w$, we have $V_w \cap H = \emptyset$. Now suppose $\langle V \cup w \rangle \cap H \neq \{0^n\}$. Then there exists a non-zero element $h \in \langle V \cup w \rangle \cap H$. 
		We assume that $h = \alpha w+v$, where {$\alpha \in \mathbb{Z}_p$} and $v \in V$. {If $\alpha = 0$, then $h \in V \cap H$. Since $V \cap H = \{0^n\}$, we have $h = 0^n$, which leads to a contradiction. Thus, $\alpha \neq 0$, i.e., $\alpha \in \mathbb{Z}_p^*$.} Since $h \in H$, we have {$\alpha^{-1}h \in H$}. Because ${\alpha^{-1}h} = w+\alpha^{-1}v \in V_w$, we have $\alpha^{-1}h \in V_w \cap H$. Thus, $\langle V \cup w\rangle \cap H \neq \{0^n\}$ implies $V_w \cap H \neq \emptyset$.
	\end{proof}
	
\section{Lower bound}
	\label{sec:lowerbound}
	In this section, we first present a lower bound by proving Theorem \ref{Th:lower}.
	\medskip
	
	\noindent
	\textbf{Theorem \ref{Th:lower}.}
	\textit{Any classical (deterministic or randomized) algorithm solving $\mathsf{GSP}(p,n,k)$ needs to make $\Omega\left(\max\{k, \sqrt{p^{n-k}}\}\right)$ queries.}
	\medskip
	
	
	
	Without loss of generality, we only need to prove the randomized query complexity, which consists of Lemma \ref{Lemma:lower1_ran} and \ref{Lemma:lower2_ran}. Before proving Lemma \ref{Lemma:lower1_ran}, we first give the results below (Claim \ref{claim:numsubgroup1} and \ref{claim:numsubgroup2}). Claim \ref{claim:numsubgroup1} is a folklore conclusion that can be found in the textbooks such as \cite{Stanley2000Enumerative}, and we give the proof for completeness. 
	\begin{Claim}
		\label{claim:numsubgroup1}
		Let $T_1$ denote the number of distinct subgroups of rank $k$ in $\mathbb{Z}_p^n$. Then $T_1 = \frac{(p^n-p^0)(p^n-p^1)\cdots(p^n-p^{k-1})}{(p^k-p^0)(p^k-p^1)\cdots(p^k-p^{k-1})}$. 
	\end{Claim}
	
	\begin{proof}
		We can specify a subgroup of rank $k$ by giving $k$ linearly independent elements in it. Now we select $k$ linearly independent elements from $\mathbb{Z}_p^n$ sequentially. Once we have selected $d$ elements, we {cannot} select the elements from the subgroup generated by these $d$ elements in the next step. Thus, we have $p^n-p^d$ possible ways to pick {the} $(d+1)$-th element. So the total number of possible ways to select $k$ linearly independent elements is {$(p^n-p^0)(p^n-p^1)\cdots(p^n-p^{k-1})$}.
		
		However, it is a double-counting process above. For a certain subgroup, there are {$(p^k-p^0)(p^k-p^1)\cdots(p^k-p^{k-1})$} ways to pick the elements generating it. Thus, the total number of subgroups is
		$$\frac{(p^n-p^0)(p^n-p^1)\cdots(p^n-p^{k-1})}{(p^k-p^0)(p^k-p^1)\cdots(p^k-p^{k-1})}.$$
	\end{proof}
	
	\begin{Claim}
		\label{claim:numsubgroup2}
		Every non-zero element in $\mathbb{Z}_p^n$ belongs to $T_2 = \frac{(p^n-p^1)(p^n-p^2)\cdots(p^n-p^{k-1})}{(p^k-p^1)(p^k-p^2)\cdots(p^k-p^{k-1})}$ distinct subgroups of rank $k$. 
	\end{Claim}
	
	\begin{proof}
		Now we count the number of subgroups of rank $k$ containing a certain non-zero element $e$. Similar to the proof of Claim \ref{claim:numsubgroup1}, we specify a subgroup of rank $k$ by giving $k$ linearly independent elements in it. We also select $k$ linearly independent elements from $\mathbb{Z}_p^n$ sequentially. Differently, we fix the first element as $e$. So the total number of possible ways to select the following $k-1$ elements is $(p^n-p^1)(p^n-p^2)\cdots(p^n-p^{k-1})$.
		
		It is also a double-counting process above. For a certain subgroup containing element $e$, there are $(p^k-p^1)(p^k-p^2)\cdots(p^k-p^{k-1})$ ways to pick the remaining $k-1$ elements to generate it. Thus, the total number of subgroups containing element $e$ is
		$$\frac{(p^n-p^1)(p^n-p^2)\cdots(p^n-p^{k-1})}{(p^k-p^1)(p^k-p^2)\cdots(p^k-p^{k-1})}.
		$$
	\end{proof}

	\begin{Lemma}
		\label{Lemma:lower1_ran}
Any classical randomized algorithm solving $\mathsf{GSP}(p,n,k)$ needs $\Omega\left(\sqrt{p^{n-k}}\right)$ queries.
	\end{Lemma}

	\begin{proof}\footnote{We are very grateful to the editors and anonymous reviewers for pointing out a flaw in the original proof  and providing the idea of the current proof. Actually, it is a generalization of the proof of the classical lower bound of Simon's problem, see \cite{simon1994power,Kitaev2002Classical,Bacon2006,Wolf2019quantum}.}	
Suppose that some randomized algorithm using at most $t$ queries output the correct answer with probability 2/3, for every possible $f$. Then it will also output the correct answer with probability 2/3 if $f$ is chosen by a randomized procedure that depends on the queries. So here
is the procedure to choose $f$. First, choose the subgroup $S$ uniformly randomly
among all subgroups of rank $k$. Then answer the queries as follows. Suppose that $x_1,...,x_{j-1}$ have been already queried, with answers $y_1,...,y_{j-1}$. If $x_i - x_{j} \in S$ for some $i < j$, then let $y_{j} = y_i$; otherwise let $y_j$ be uniformly
distributed in $X \setminus \{y_1,...,y_{j-1}\}$. In the following, it suffices to prove any randomized algorithm using at most $c\cdot\sqrt{p^{n-k}}$ queries cannot give the correct answer with probability $2/3$ if $f$ is chosen by the above procedure for some constant $c$.

Now consider an arbitrary sequence of queries $x_1,...,x_{j}$ ($j \le \sqrt{p^{n-k}/2}$), which we can assume to be all distinct. We say it is \textit{good} if it shows a collision (i.e. $y_i = y_{i'}$ for some $i \neq i'$); otherwise, it is \textit{bad}. 
Let $\mathcal{H}$ denote the set of the subgroups of rank $k$. Then $|\mathcal{H}| = T_1$ as Claim  \ref{claim:numsubgroup1}. Let $\mathcal{H}_j$ denote the set of the subgroups of rank $k$ that do not contain $x_i-x_l$ for any $i<l \le j$. 
Since $x_i-x_l$ belongs to $T_2$ subgroups of rank $k$ by Claim \ref{claim:numsubgroup2} for any $i < l \le j$, we have 
$|\mathcal{H}_j| \ge T_1 - \binom{j}{2}T_2$. 

Next, we will prove if $x_1,...,x_{j}$ is bad, there are still $|\mathcal{H}_{j}|$ equally possible values for $S$. For any $H \notin \mathcal{H}_{j}$, $Pr\{x_1,...,x_{j}\text{ is bad}|S = H\} = 0$. For any $H \in \mathcal{H}_{j}$, $Pr\{x_1,...,x_{j}\text{ is bad}|S = H\} = 1$. Thus, if $H \in \mathcal{H}_{j}$, then
$$
\begin{aligned}
&Pr\{S = H|x_1,...,x_{j}\text{ is bad}\} \\
&= \frac{Pr\{S=H\}Pr\{x_1,...,x_{j}\text{ is bad}|S=H\}}{\sum_{H' \in \mathcal{H}}Pr\{S=H'\}Pr\{x_1,...,x_{j}\text{ is bad}|S = H'\}} \\
&= \frac{Pr\{S=H\}}{\sum_{H' \in \mathcal{H}_{j}}Pr\{S = H'\}} \\ 
&=\frac{1}{|\mathcal{H}_{j}|}.
\end{aligned}
$$

For any $i < j$, since $x_i-x_j$ belong to $T_2$ subgroups of rank $k$, we have $Pr\{x_i-x_j \in S| x_1,...,x_{j-1} \text{ is bad}\} \le \frac{T_2}{|\mathcal{H}_{j-1}|}$. Thus, 
$$
Pr\{x_1,...,x_j \text{ is good }|x_1,...,x_{j-1} \text{ is bad}\} \le 
(j-1)\frac{T_2}{|\mathcal{H}_{j-1}|}.
$$
Therefore, for an arbitrary sequence of $t$ queries $x_1,...,x_t$ ($t \le \sqrt{p^{n-k}/2}$), the probability that $x_1,...,x_t$ is good is 
$$
\begin{aligned}
P_t &= \sum_{j=2}^t Pr\{x_1,...,x_{j-1} \text{ is bad}, x_1,...,x_j \text{ is good}\} \\
&\le \sum_{j=2}^t Pr\{x_1,...,x_j\text{ is good}| x_1,...,x_{j-1}\text{ is bad}\} \\
&\le \sum_{j=2}^t(j-1)\frac{T_2}{|\mathcal{H}_{j-1}|} \\
&< \binom{t}{2}\frac{T_2}{|\mathcal{H}_{t}|}.
\end{aligned}
$$
Since $t < \sqrt{p^{n-k}/2}$, we have
$$
\binom{t}{2}\frac{T_2}{T_1} < \frac{p^{n-k}}{4}\frac{p^k-1}{p^n-1} < \frac{1}{4}.
$$
Thus,
$$
|\mathcal{H}_t| \ge T_1-\binom{t}{2}T_2 = T_1\left(1-\binom{t}{2}\frac{T_2}{T_1}\right) > \frac{3}{4}T_1, 
$$
and 
$$
P_t < \frac{4}{3}\cdot\binom{t}{2}\frac{T_2}{T_1}  < \frac{1}{3}.
$$
If $x_1,...,x_t$ is bad, then algorithm must choose one subgroup out of the last $|\mathcal{H}_t|$ equally possible subgroups. When $n$ is enough large (i.e. $n \ge 3$), 
$$
T_1 > p^{(n-k)k} \ge p^{n-1} \ge 4,
$$
and thus $|\mathcal{H}_t|  > \frac{3}{4}T_1> 3$. Therefore, the success probability of an arbitrary algorithm using at most $\sqrt{p^{n-k}/2}$ queries is less than $1/3+1/3 = 2/3$.
	\end{proof}

	In the following proof, if $Q$ is the set of queried elements in an algorithm to solve $\mathsf{GSP}(p,n,k)$, let $T_Q = \{q_1- q_2\ |\ q_1,q_2 \in Q\}$. It is easy to see $|T_Q| \le |Q|^2$ and $0^n \in T_Q$.
	\begin{Lemma}
		\label{Lemma:lower2_ran}
		Any classical randomized algorithm solving $\mathsf{GSP}(p,n,k)$ needs $\Omega\left(k\right)$
		queries.
	\end{Lemma}
	\begin{proof}
	  Suppose an algorithm makes no more than $k/2$ queries. Denote the set of queried elements by $Q$. Let $S_1 = \langle T_Q \rangle, S_2= \langle T_Q \cap S \rangle, x = rank(S_2), y = rank(S_1)$. Since $S_2 \le S_1$, we have $x \le y$. Since $Q$ contains no more than $k/2$ queries, $T_Q$ contains at most $k/2$ linear independent elements, i.e, $y \le k/2$. Thus, $x \le y \le k/2$. 
	   Now we count the number of possible target subgroups until now. We have found a subgroup of rank $x$ of $S$. Thus, we only need to find other $k-x$ linear independent elements in $S$ in turn. In the first step, we try to find first element $e$. There are at least $p^{n}-p^y$ candidate elements, i.e., those elements not in $S_1$. In the next step, there are at least $p^{n}-p^{y+1}$ candidate elements, i.e., those elements not in $\langle S_1 \cup e \rangle$. In the end, we have $(p^{n}-p^y)(p^{n}-p^{y+1})\cdots(p^{n}-p^{y+(k-x+1)})$ ways to select these $k-x$ elements. 
	   However, it is a double-counting process above. For a certain subgroup of rank $k$, if we have known a subgroup of rank $x$ of it, there are $(p^{k}-p^x)(p^{k}-p^{x+1})\cdots(p^{k}-p^{k-1})$ ways to pick other $k-x$ linear independent elements. Thus, the total number of subgroups is no less than $$\frac{(p^{n}-p^y)(p^{n}-p^{y+1})\cdots(p^{n}-p^{y+(k-x+1)})}{(p^{k}-p^x)(p^{k}-p^{x+1})\cdots(p^{k}-p^{k-1})} > (\frac{p^{n}-p^{y+k-x-1}}{p^{k}-p^{k-1}})^{k-x} > p^{(n-k)k/2},$$
	  which means the success probability of the algorithm is less than $\frac{1}{p^{(n-k)k/2}}$. Thus, $k/2$ queries is not enough to find the target subgroup with a high probabilitiy.

	\end{proof}

	In the following, we give Claim \ref{Claim:numjoin1} first. Then we present a lower bound for the nonadaptive deterministic query complexity by proving Theorem \ref{Th:lower_non}.

	\begin{Claim}
		\label{Claim:numjoin1}
		For any set $D \subseteq \mathbb{Z}_p^n$ such that $0^n \notin D$ and $|D| < \frac{d(p^n-1)}{(p^{k}-1)} (1 \le d \le k \le n)$, there exists a subgroup $S_1$ of rank $k$ of $\mathbb{Z}_p^n$ such that $|S_1 \cap D| < d$.
	\end{Claim}
	
	\begin{proof}
		We prove the result by contradiction. Suppose for any group $S_1$ of rank $k$, there exist at least $d$ elements in $S_1 \cap D$. 
		Since $\mathbb{Z}_p^n$ has $T_1$ distinct subgroups of rank $k$ by Claim \ref{claim:numsubgroup1}, the total number of elements is $T_1\cdot d$. However, some elements may be counted repeatedly. Since every non-zero element in $\mathbb{Z}_p^n$ belongs to $T_2$ distinct subgroups of rank $k$ by Claim \ref{claim:numsubgroup2}, we have 
		$T_2\cdot|D| \ge {T_1}\cdot d$,
		which means $|D| \ge \frac{dT_1}{T_2} = \frac{d(p^n-1)}{(p^{k}-1)}$.
		It leads to a contradiction. Thus, there exists a subgroup $S_1$ of rank $k$ such that $|S_1 \cap D| < d$. 
	\end{proof}

	\medskip
	
	\noindent\textbf{Theorem \ref{Th:lower_non}.}
	\textit{Any nonadaptive classical deterministic algorithm solving $\mathsf{GSP}(p,n,k)$ needs $\Omega\left(\max\{k,\sqrt{k\cdot p^{n-k}}\}\right)$ queries.}
	\medskip

	\begin{proof}
		It is easy to see any nonadaptive classical deterministic algorithm solving $\mathsf{GSP}(p,n,k)$ needs $\Omega(k)$ queries by Lemma \ref{Lemma:lower2_ran}. Now we prove the other part. Suppose there exists an nonadaptive algorithm $\mathcal{A}$ which makes less than $\sqrt{\frac{k(p^n-1)}{(p^{k+1}-1)}}$ queries to solve $\mathsf{GSP}(p,n,k)$ correctly. Suppose the set of queried element in $\mathcal{A}$ is $Q$. Then $|Q| < \sqrt{\frac{k(p^n-1)}{(p^{k+1}-1)}}$, so we have $|T_Q| < \frac{k(p^n-1)}{(p^{k+1}-1)}$. By Claim \ref{Claim:numjoin1}, there exists a subgroup $\hat{S}$ of rank $k+1$ such that the number of non-zero elements in $\hat{S} \cap T_Q$ is less than $k$. Let $\hat{S'} = \langle \hat{S} \cap T_Q \rangle$. Then $rank(\hat{S'}) < k$. 
		Let $S_1, S_2$ be two different subgroups of rank $k$ of $\mathbb{Z}_p^n$ such that $\hat{S'} < S_1,S_2 < \hat{S}$.
	
		Now we pick a function $f_1$ satisfying that $f_1(x) = f_1(y)$ iff $x - y \in S_1$ as the input of Algorithm $\mathcal{A}$. Since Algorithm $\mathcal{A}$ solves the problem correctly, its output is $S_1$. Next we construct a function $f_2$ as the input of Algorithm $\mathcal{A}$, where $f_2$ satisfies that (i) $\forall x\in Q, f_2(x) = f_1(x)$; 
		(ii) $f_2(x) = f_2(y)$ iff $x - y \in S_2$. Conditions (i) and (ii) mean for any $x,y\in Q$, $f_2(x) = f_2(y)$ iff $x - y \in S_1\cap T_Q$. Since $S_1\cap T_Q \subseteq \hat{S}\cap T_Q \subseteq \hat{S'} \subseteq S_2$, conditions (i) and (ii) are simultaneously satisfiable. Condition (i) implies the output of Algorithm $\mathcal{A}$ is still $S_1$, which is a wrong answer. Thus, Algorithm $\mathcal{A}$ does not solve the problem correctly for any input, which leads to a contradiction. Therefore, any classical deterministic algorithm solving $\mathsf{GSP}(p,n,k)$ needs at least $\sqrt{\frac{k(p^n-1)}{(p^{k+1}-1)}}$ queries, i.e., 
		we need to make $\Omega\left(\sqrt{k\cdot p^{n-k}}\right)$ queries.
	\end{proof}

	\section{Algorithm and upper bound}
	\label{sec:upperbound}
	In this section, we propose an algorithm to solve $\mathsf{GSP}(p,n,k)$ and analyze its query complexity, which establishes the upper bound in Theorem \ref{Th:upper}. 
	\medskip
	
	\noindent
	\textbf{Theorem \ref{Th:upper}.}
	\textit{There exists a classical deterministic algorithm to solve $\mathsf{GSP}(p,n,k)$ using $O\left(\max\{k,\sqrt{k\cdot p^{n-k}}\}\right)$ queries.
	}
	\medskip
	
	
	For our problem, we say $x$ and $y$ collide iff $f(x) = f(y)$, which means $x - y \in S$. Thus, $x$ and $y$ collide if $x$ and $y$ belong to the same coset of $S$. So the probability of the collision in each pair of elements is $\frac{1}{p^{n-k}}$. Once we find $k$ linearly independent elements in $S$, we will determine what $S$ is. Thus, there exists a randomized algorithm in intuition as follows. By making $O(\sqrt{k \cdot p^{n-k}})$ queries randomly and uniformly, the expected number of collisions is $O(k p^{n-k}\frac{1}{p^{n-k}})=O(k)$, from which it is likely to find a generating set of $S$ with a high probability. 
	
	This randomized algorithm gives an intuitive guess on the upper bound of query complexity. However, it seems hard to derandomize it directly. Therefore, we give an ingenious construction of the query set to ensure that a generating set of $S$ can be found in our main algorithm (i.e., Algorithm \ref{al:main}), which is based on group theory and a divide-and-conquer subroutine (i.e., Algorithm \ref{al:findGroup}). In the following, we first introduce Algorithm \ref{al:findGroup} and analyze its correctness and query complexity. Then we describe and analyze Algorithm \ref{al:main}.	
	\subsection{Critical subroutine}
	In Algorithm \ref{al:findGroup}, given a subgroup $A$ satisfying that all elements of $A$ have been queried and $A \cap S = \{0^n\}$, a subgroup $S_1 \le S$, and an integer $d$, we wish to find a subgroup $B$ of rank $d$ such that $A \cap B = \{0^n\}$ and $(A + B)\cap S = \{0^n\}$. Also, we query all the elements in $B$ having not been queried. Meanwhile, we find a subgroup $S_2$ such that $S_1 \le S_2 \le S$, which is an extra bonus. 
	
	We use the idea of recursion. First, we find a subgroup $B'$ of rank $d-1$ such that $A \cap {B'} = \{0^n\}, (A + B') \cap S = \{0^n\}$ and a subgroup $S_2$ such that $S_1 \le S_2 \le S$. We also query all the elements in $B'$. Next, we query an element $u \in \mathbb{Z}_p^n \setminus(S_2 + A + B')$. If there exists an element $b \in B'$ such that $f(b) = f(u)$, then we have $b- u \in S$. Thus, we expand $S_2$, and query an element $u \in \mathbb{Z}_p^n \setminus(S_2 + A + B')$ again. We repeat this procedure until for any $b \in B'$ we have $f(b) \neq f(u)$. Then we query all the elements in $\langle B'\cup u \rangle \setminus(B' \cup u)$, i.e., the elements not queried in $\langle B'\cup u \rangle$. Similarly, if there exists $a\in A$, $b\in \langle B'\cup u \rangle\setminus \{0^n\}$ such that $f(a) = f(b)$, then we have {$a - b \in S$}. So we expand $S_2$ and go back to Step 6. Only if for any $a\in A, b\in \langle B'\cup u \rangle\setminus \{0^n\}$ we have $f(a) \neq f(b)$, we exit the outer loop. Then we obtain $B = \langle B'\cup u \rangle$ and return $(B,S_2)$. At this point, all the elements in $B$ have been queried.
	\begin{algorithm}
		\caption{findGroup}
		\label{al:findGroup}
		\begin{algorithmic}[1]
			\REQUIRE $A,S_1,d$ s.t. $0 \le d \le n-k,S_1 \le S, A \cap S = \{0^n\}$, {all elements of $A$ have been queried};
			\ENSURE $B,S_2$ s.t. $rand(B) = d, A \cap B = \{0^n\},(A + B) \cap S = \{0^n\}$, {$S_1 \le S_2 \le S$, all elements of $B$ have been queried};
			\IF {d = 0}
			\RETURN $(\{0^n\}, S_1)$; 
			\ENDIF
			\STATE $(B',{S_2}) \leftarrow findGroup(A,S_1,d-1);$
			\REPEAT
			\REPEAT
			\STATE Query a non-zero element $u$ s.t. $u \in \mathbb{Z}_p^n$ and $u \notin ({S_2}+ A + B')$;
			\IF {$\exists b \in B'$ s.t. $f(b) = f(u)$}
			\STATE ${S_2} \leftarrow \langle {S_2} \cup (b - u) \rangle$;
			\ENDIF
			\UNTIL{$\forall b \in B'$ s.t. $f(b) \neq f(u)$}
			\STATE Query all the elements in $\langle B'\cup u\rangle \setminus(B' \cup u)$;
			\IF {$\exists a \in A,b\in \langle B'\cup u \rangle \setminus \{0^n\}$ s.t. $f(a) = f(b)$}
			\STATE ${S_2} \leftarrow \langle {S_2} \cup (a - b) \rangle$;
			\ENDIF
			\UNTIL{$\forall a \in A,b\in \langle B'\cup u \rangle \setminus \{0^n\}$ s.t. $f(a) \neq f(b)$}
			\STATE $B \leftarrow \langle B' \cup u \rangle$; 
			\RETURN $(B, {S_2})$;
		\end{algorithmic}
	\end{algorithm}
	
	We show that the algorithm obtains the desired output below. Since $u \notin (A+B')$, we have $A \cap B'_u = \emptyset$, which implies $A \cap \langle B' \cup u \rangle = \{0^n\}$, i.e., $A \cap B = \{0^n\}$ by Fact \ref{Fact:join}.	Since $f(a) \neq f(b)$ implies $a-b \notin S$ for any $a\in A, b\in B\setminus \{0^n\}$, we have $((A+B)\setminus\{0^n\}) \cap S = \emptyset$. In addition, because $0^n \in S$, we have $(A+B) \cap S = {\{0^n\}}$. Moreover, while we expand $S_2$, $S_2$ is always a subgroup of $S$. Thus, we obtain desired $B$ and $S_2$.
	
	
	Now, we analyze the query complexity of Algorithm \ref{al:findGroup} by proving Lemma \ref{Lemma:querytime}. 
	
	\begin{Lemma}
		\label{Lemma:querytime}
		In Algorithm \ref{al:findGroup}, suppose $rank(S_i) = k_i$ for $i=1,2$. If $A\ {\neq}\ \{0^n\}$, the number of queries is $p^d-1+(k_2-k_1)\cdot(p^d-p^{d-1})$. If $A\ {=}\ \{0^n\}$, the number of queries is at most $p^d-1+(k_2-k_1)$.
	\end{Lemma}
	
	\begin{proof}
		We consider the following two cases: $A \neq \{0^n\}$ and $A = \{0^n\}$. 
		
		Case 1: $A \neq \{0^n\}$. In Step 7 and 12, we make $p^{d}-p^{d-1}$ queries to query all the element in $\langle B'\cup u \rangle \setminus B'$. Suppose we never execute Step 9 and 14. When we recursively call the algorithm itself, the number of queries is $(p^d-p^{d-1})+(p^{d-1}-p^{d-2})+\cdots+(p^1-p^0)= p^d-1$ in total. However, once we execute Step 9, we will go back to the inner loop and make one extra query. Once we execute Step 14, we will go back to the outer loop and make at most $p^d-p^{d-1}$ extra queries. Moreover, every time we execute Step 9 or 14, the rank of $S_2$ adds one. Since we input $S_1$ and obtain $S_2$, the total number {of times} we execute Step 9 and 14 is $k_2-k_1$. Thus, we need at most extra $(k_2-k_1)\cdot(p^d-p^{d-1})$ queries in order to expand $S_1$ to $S_2$. Finally, the total number of queries is at most $p^d-1+(k_2-k_1)\cdot(p^d-p^{d-1})$. 
		
		Case 2: $A = \{0^n\}$. Different from Case 1, we never execute Step 14 in this case. Suppose we go into Step 13. At this time, we have $f(b) \neq f(u)$ for any $b \in B'$. Thus {$u - b \notin S$}. Since $u-B'=u+B'=B'+u=B'_u$, we have $B'_u \cap S = \emptyset$, which implies $\langle B' \cup u \rangle \cap S = \{0^n\}$. Hence, we have $f(a) \neq f(b)$ for any $a \in \{0^n\}, b\in \langle B' \cup u \rangle\setminus \{0^n\}$. Therefore, the if-statement is always judged to be false, and then we jump out of Step 14. On the other hand, similar to Case 1, if we never execute Step 9, then the number of queries is $p^d-1$. Once we execute Step 9, we will go back to the inner loop and make one extra query. Thus, we need extra $k_2-k_1$ queries in order to expand $S_1$ to $S_2$. Finally, the total number of queries is $p^d-1+(k_2-k_1)$.
	\end{proof}

	\subsection{Main algorithm}
	In this section, we introduce and analyze Algorithm \ref{al:main}. The main idea of Algorithm \ref{al:main} is as follows. First, we try to obtain two subgroups $A,B$ such that the elements between $A$ and $B$ do not collide by calling Algorithm \ref{al:findGroup}. We may find some collisions in this process. Second, we try to find some other collisions between $B$ and cosets of $A$. By the above collisions, we make sure that a generating set of $S$ will be found. 
		\begin{algorithm}
		\caption{Find $S$}
		\label{al:main}
		\begin{algorithmic}[1]
			\REQUIRE $d \in \{0,...,n-k\}$;
			\ENSURE subgroup $S$;
			\STATE {Query $0^n$;}
			\STATE $(A,S_1) \gets findGroup(\{0^n\},\{0^n\},n-k-d)$;
			\STATE $(B,S_2) \gets findGroup(A,S_1,d)$;
			\STATE $V \leftarrow A + B$, $W \leftarrow \overline{ V + S_2 }$, find a basis $\{w_i\}$ of $W$;
			\FOR{$i = 1\ to\ rank(W)$}  
			\STATE Query all the elements in $B_{w_i}$, and then we can find a pair of $(a_i,b_i)$ s.t.\ {$a_i \in A, b_i \in B_{w_i}$} and $f(a_i) = f(b_i)$;
			\STATE $s_{w_i} \leftarrow b_i-a_i$;
			\ENDFOR
			\RETURN $\langle S_2 \cup \{s_{w_i}\} \rangle$;
		\end{algorithmic}
	\end{algorithm}
	

	Whatever the value of $d$ is, Algorithm \ref{al:main} will find $S$, but there exists an optimal value such that the number of queries used is minimum. The details will be discussed later. Now 
	we describe Algorithm \ref{al:main} in detail and show its correctness. Let $rank(S_i) = k_i$ for $i=1,2$. {In Step 1, we query $0^n$ to obtain the value of $f(0^n)$, which will be used in Step 2.} In Step 2, we find a subgroup $A$ such that $A\cap S = \{0^n\}$ and $rank(A) = n-k-d$ by calling Algorithm \ref{al:findGroup}. At the same time, we find a subgroup $S_1 \le S$. In Step 3, we find a subgroup $B$ such that $A \cap B = \{0^n\}$, $(A + B)\cap S = \{0^n\}$, and $rank(B) = d$ by calling Algorithm \ref{al:findGroup} again. Simultaneously, we find a subgroup $S_2$ such that $S_1 \le S_2 \le S$. In Step 4, we obtain $W$, a complement of $V+ S_2$ in $\mathbb{Z}_p^n$. Let $\{w_i\}$ be a basis of $W$. In Step 5-8, we query all the elements in $B_{w_i}$ to obtain $s_{w_i}$ for any $i$. Next we elaborate why we can obtain the desired $\{s_{w_i}\}$. Suppose $f(a) \neq f(b)$ for any {$a \in A, b \in B_{w_i}$} in Step 6. Then $b - a \notin S$ for any {$a \in A, b \in B_{w_i}$}. Since $V_{w_i} = V + w_i = B+A+w_i = B-A+w_i = B+w_i-A = \{b-a\ |\ a\in A, b\in B_{w_i}\}$, we have $ V_{w_i} \cap S = \emptyset$. However, Lemma \ref{Lemma:coset} implies that $V_{w_i} \cap S \neq \emptyset$, which leads to a contradiction. Thus, there exists $a_i \in A, b_i \in B_{w_i}$ such that $f(b_i) = f(a_i)$. Then $s_{w_i} = b_i - a_i$ is a non-zero element in $V_{w_i} \cap S$. Finally, Lemma \ref{Lemma:findS} implies $S_2 \cup \{s_{w_i}\}$ is a generating set of $S$. Thus, we find $S$.
	
	\begin{Lemma}
		\label{Lemma:coset}
		Suppose $V$ is a subgroup of $\mathbb{Z}_p^n$ such that $V \cap S = \{0^n\}$ and $rank(V) = n-k$. For $w \notin V$, we have $V_w \cap S \neq \emptyset$.
	\end{Lemma}
	
	\begin{proof}
		For $w\notin V$, $\langle V \cup w \rangle$ is a group of rank $n-k+1$. First we prove $\langle V \cup w \rangle \cap S \neq \{0^n\}$ by contradiction. Suppose $\langle V \cup w \rangle \cap S = \{0^n\}$. Then we have $rank(\langle V \cup w \cup S\rangle) = n-k+1+k = n+1$. Since $\langle V \cup w \cup S \rangle$ is a subgroup of $\mathbb{Z}_p^n$, we have $rank(\langle V \cup w \cup S \rangle) \le n$, which leads to a contradiction. Thus, we have $\langle V \cup w \rangle \cap S \neq \{0^n\}$, which implies $V_w \cap S \neq \emptyset$ by Fact \ref{Fact:join}.

	\end{proof}
	
	\begin{Lemma}
		\label{Lemma:findS}
		In Algorithm \ref{al:main}, $S_2 \cup \{s_{w_i}\}$ is a generating set of $S$.
	\end{Lemma}
	
	\begin{proof}
		Let $\{v_i\}$ be a basis of $V$, $\{s_i\}$ be a basis of $S_2$. Since $V \cap S_2 = \{0^n\}$, we have $rank(V+S_2) = rank(V)+rank(S_2) = n-k+k_2$. Since $W$ is a complement subgroup of $V + S_2$, we have $rank(W) = n-(n-k+k_2) = k-k_2$ and $\{u_i\} = \{v_i\} \cup \{s_i\} \cup \{w_i\}$ constructs a basis of $\mathbb{Z}_p^n$ by Fact \ref{Fact:exgroup}. Since $s_{w_i} \in V_{w_i}$, we have
		$$
		s_{w_i} = \sum_{j=1}^{n-k}\alpha_{ij}v_j+w_i.
		$$
		Suppose we have                                                                             
		$$
		\sum_{i=1}^{k_2}\beta_i s_i + \sum_{i=1}^{k-k_2}\gamma_i s_{w_i} = 0.
		$$
		That is, 
		$$
		\sum_{i=1}^{k_2}\beta_i s_i+\sum_{i=1}^{k-k_2}\sum_{j=1}^{n-k} \gamma_i \alpha_{ij}v_j +
		\sum_{i=1}^{k-k_2}\gamma_i w_i = 0, 
		$$
		which implies 
		$\beta_i = 0$ for any $1 \le i \le k_2$, and $\gamma_i = 0$ for any $1 \le i \le k-k_2$. Thus,
		$\{s_{w_i}\}\cup \{s_i\}$ is linearly independent by Fact \ref{Fact:linear}. Since $rank(S) = k$, the cardinality of a linearly independent subset of {$S$} is at most $k$. Thus, $|\{s_{w_i}\}\cup \{s_i\}| = k$ implies $\{s_{w_i}\}\cup \{s_i\}$ consists of a basis of $S$. Since $\{s_i\} \subseteq S_2$, $S_2 \cup \{s_{w_i}\}$ constructs a generating set of $S$.
	\end{proof}
	
	Now we analyze the query complexity of Algorithm \ref{al:main}. {The number of queries in Step 1 is 1.} Lemma \ref{Lemma:querytime} implies that the number of queries in Step 2 is $p^{n-k-d}-1+k_1$, and the number of queries in Step 3 is at most $p^{d}-1+(k_2-k_1)\cdot (p^d-p^{d-1})$. In addition, the number of queries in Step 5-8 is $(k-k_2)\cdot p^d$. Let $TQ$ denote the total number of queries. Then
	$$
	\begin{aligned}
	TQ &\le {1\ +\ } p^{n-k-d}-1+k_1+p^{d}-1+(k_2-k_1)\cdot (p^d-p^{d-1})+(k-k_2)\cdot p^d \\
	&\le p^{n-k-d}+(k+1)\cdot p^d.
	\end{aligned}
	$$
	If $n \ge k+ \log_p k$, let $d = (n-k-\log_p k)/2$. Then we get
	$$ TQ \le \sqrt{k\cdot p^{n-k}}+\frac{k+1}{\sqrt{k}}\sqrt{p^{n-k}} = O(\sqrt{k \cdot p^{n-k}}). $$
	Otherwise, let $d = 0$, we have
	$$ TQ \le p^{n-k}+k+1 < 2k+1 = O(k). $$
	That is, the number of queries of our algorithm is 
	$$
	TQ = 
	\left\{
	\begin{array}{lrr}
	O(\sqrt{k \cdot p^{n-k}}) & n \ge k+ \log_p k& \\
	O(k) & n < k+ \log_p k&  
	\end{array}
	\right..
	$$
	It is equivalent to
	$$
	TQ = O(\max\{k, \sqrt{k \cdot p^{n-k}}\}),
	$$
	which meets the upper bound in Theorem \ref{Th:upper}.

	\section{Conclusion}
	\label{sec:conclusion}
	In this paper, we obtain an almost full characterization of the classical query complexity for the generalized Simon's problem $\mathsf{GSP}(p,n,k)$. We prove that any classical (deterministic or randomized) algorithm has to make $\Omega\left(\max\{k, \sqrt{p^{n-k}}\}\right)$ queries. Also, we prove that any classical nonadaptive deterministic algorithm has to make $\Omega\left(\max\{k, \sqrt{k\cdot p^{n-k}}\}\right)$ queries. Moreover,
	we devise {a} deterministic algorithm for this problem with query complexity of $O\left(\max\{k, \sqrt{k \cdot p^{n-k}}\}\right)$.
	These results clarify the computational power gap between quantum and classical computing on this problem. Furthermore, two open problems remain: 
	\begin{itemize}
	\item{It still has a gap between the lower bound and the upper bound for {adaptive} classical query complexity of the generalized Simon's problem.} 
	
	\item{If we generalize the domain of $f$ from $\mathbb{Z}_p^n$ to an arbitrary finite Abelian group $\mathbb{Z}_{m_1} \oplus \mathbb{Z}_{m_2} \oplus \cdots \oplus \mathbb{Z}_{m_n}$, where $\mathbb{Z}_{m_i}$ denotes the additive cyclic group of order $m_i$, what is the optimal classical query complexity in this case?}
	\end{itemize}
	
	\section*{Declaration of competing interest}
	The authors declare that they have no known competing financial interests or personal relationships that could have
appeared to influence the work reported in this paper.
	
	\section*{Acknowledgment}
	The authors are very grateful to the editors and anonymous reviewers for their helpful comments, and particularly   for pointing out a flaw in the original proof of Lemma \ref{Lemma:lower1_ran} and providing us with a simpler proof idea. This work was supported by the National Natural Science Foundation of China (Grant No. 61772565), the Guangdong Basic and Applied Basic Research Foundation (Grant No. 2020B1515020050), the Key Research and Development project of Guangdong Province (Grant No. 2018B030325001).
	
	\bibliography{GSP}

\newpage

\begin{appendix}
\section{An exact quantum algorithm for generalized Simon's problem}
\label{exact quantum algorithm}
Brassard and H\o yer \cite{brassard1997exact} provided an algorithm to solve Simon's problem with $O(n)$ queries exactly. In this section, we review the algorithm and make some small adjustments to the algorithm to solve $\mathsf{GSP}(p,n,k)$. Particularly, we replace the method of improving the success probability in \cite{brassard1997exact} with a more general method, amplitude amplification \cite{Brassard2002quantum}, such that the algorithm needs to only make $O(n-k)$ queries. 

We first give some notation. Let $G = \mathbb{Z}_p^n$. {In $\mathsf{GSP}(p,n,k)$, without loss of generality, suppose the elements in $X$ are represented by $\lceil n\log p \rceil$-bit strings. Quantum query oracle $O_f$ is defined as $O_f \ket{g}\ket{b} = \ket{g}\ket{f(g)\oplus b}$ for any $g \in G$, $b \in X$, where $\oplus$ denotes bitwise addition modulo 2.} The query complexity of a quantum algorithm is the number of $O_f$ used in the algorithm. Let $T_0$ be a transversal of $S$ in $G$, i.e, $T_0$ consists of exactly one representative {from} each coset of $S$. For $g,h \in G$, let
$$
g \cdot h = (\sum_{i=1}^n g_i h_i) \mod p,
$$
where $g = (g_1,...,g_n)$ and $h = (h_1,...,h_n)$. For $H \le G$, let
$$
H^{\perp} = \{g \in G| g \cdot h =0 \text{ for all }h \in H\}
$$
denote the \textit{orthogonal subgroup} of $H$. Let $$\mu(g,h) = e^{\frac{2\pi i}{p}{g\cdot h}},$$ 
$$F_G = \frac{1}{\sqrt{|G|}}\sum_{g,h\in G}\mu(g,h)|g\rangle \langle h|,$$ $$\phi_h = \sum_{g \in G}\mu(h,g)|g\rangle \langle g|.$$ 
For $X \subseteq G$ and $g \in G$, let $$\ket{\phi_gX} = \frac{1}{{\sqrt{|X|}}}\sum_{x \in X}\mu(g,x)\ket{x}.$$  

Next, we introduce the algorithm process. It is worth noting that the idea of the quantum algorithm is to find a basis of $S^{\perp}$. Then $S = {(S^{\perp})}^{\perp}$. We give Algorithm \ref{sub1} first, which is a subroutine in Simon's algorithm and generalized to the case of $\mathbb{Z}_p$. The goal of Algorithm \ref{sub1} is to obtain $\ket{\Psi}$. When we measure the first register of $\ket{\Psi}$, we can get an element in $S^{\perp}$ with certain.

\begin{algorithm}
	\caption{Simon subroutine}
	\label{sub1}
	\textbf{Input:} Quantum state $\ket{0^n}\ket{0}$;
	
	\textbf{Procedure:}
	\begin{enumerate}

		\item Apply the inverse of transform $F_G$ to the first register producing an equally-weighted superposition of all elements in the group $G$:
		$$
		\frac{1}{\sqrt{|G|}}\sum_{g \in G} \ket{g}\ket{0}.
		$$
		
		\item Apply $O_f$, producing a superposition of all cosets of {$S$}:
		$$
		\frac{1}{{\sqrt{|G|}}}\sum_{g \in G} \ket{g}\ket{f(g)} = \frac{1}{\sqrt{T_0}}\sum_{t \in T_0}\ket{t + S}\ket{f(t)}.
		$$
		
		\item Apply $F_G$ to the first register, producing a superposition over the orthogonal subgroup $S^{\perp}$:
		$$
		\ket{\Psi} = \frac{1}{{\sqrt{|T_0|}}}\sum_{t \in T_0}\ket{\phi_t S^{\perp}}\ket{f(t)}.
		$$
	\end{enumerate}
	\textbf{Output:} Quantum state $\ket{\Psi}$.
\end{algorithm}

Suppose we have obtained some element $y$ in $S^{\perp}$, we do not wish to obtain it again. Instead, we wish to find a basis of $S^{\perp}$ as soon as possible. Thus, Algorithm \ref{sub2} is proposed to shrink the range of superposition and avoiding obtaining the same elements. 

\begin{algorithm}
	\caption{Shrinking a subgroup}
	\label{sub2}
	\textbf{Input:} Quantum state $\ket{\phi_g H}\ket{0}$; $H \le G$ be a nontrivial subgroup; $y \in H/\{0^n\}$; $j$ such that $y_j \neq 0$ and $K = \{h \in H|h_j = 0\}$.
	
	\textbf{Procedure:}
	\begin{enumerate}
		\item Apply operation $U_1: U_1\ket{x}\ket{z} = \ket{x}\ket{x+z}$ to the $j$th qubit in the first register and the second register of $\ket{\phi_g H}\ket{0}$, producing the state:
		
		\begin{align*}
		&\frac{1}{\sqrt{|H|}}\sum_{h \in H}\mu(g,h)\ket{h}\ket{h_j} \\ = &\frac{1}{\sqrt{p}}\sum_{i \in \mathbb{Z}_p}\mu(g,iy)\left(\frac{1}{\sqrt{|K|}}\sum_{k \in K}\mu(g,k)\ket{iy+k}\right)\ket{i}.
		\end{align*}

		\item Apply operator $U_2:$ {$U_2\ket{x}\ket{i} = \ket{x-iy}\ket{i}$} to the above state. This produces
		\begin{align*}
	  &\frac{1}{\sqrt{p}}\sum_{i \in \mathbb{Z}_p}\mu(g,iy)\left(\frac{1}{\sqrt{|K|}}\sum_{k \in K}\mu(g,k)\ket{k}\right)\ket{i} \\
	  = & \ket{\phi_g K}\left(\frac{1}{\sqrt{p}}\sum_{i \in \mathbb{Z}_p}\mu(g,iy)\ket{i}\right).
		\end{align*}

		\item Apply {$F^{\dagger}_{\mathbb{Z}_p}$: $F^{\dagger}_{\mathbb{Z}_p}\ket{i}= \frac{1}{\sqrt{p}}\sum_{a\in \mathbb{Z}_p}\mu(i,-a)\ket{a}$} to the second register, producing a superposition: 
		$$
		\ket{\phi_g K}\ket{g \cdot y}.
		$$
	\end{enumerate}
	\textbf{Output:} Quantum state $\ket{\phi_g K}\ket{g \cdot y}$.
\end{algorithm}

Suppose $Y = \{y_1, ..., y_m\}$ is a known linearly independent set in $S^{\perp}$ ($0 \le m < n-k$). Then there exists a subgroup $K_m \le S^{\perp}$ such that $S^{\perp} = K_m + \langle Y \rangle$ and a quantum routine $\mathcal{A}_m$ such that $\mathcal{A}_m \ket{\phi_g S^{\perp}}\ket{0^m} = \ket{\phi_g K_m}\ket{g \cdot y_1, ..., g \cdot y_m}$ by applying Algorithm \ref{sub2} repeatedly. Moreover, we have a quantum routine $\mathcal{A}'_m$ such that $\mathcal{A}'_m\ket{0^n}\ket{0}\ket{0^m} = \frac{1}{\sqrt{|T_0|}}\sum_{t \in T_0}\ket{\phi_t K_m}\ket{f(t)}\ket{t \cdot y_1,...,t \cdot y_m}$ by applying Algorithm \ref{sub1} and $\mathcal{A}_m$. Algorithm \ref{sub1} makes one query, $\mathcal{A}_m$ does not make queries, thus $\mathcal{A}'_m$ makes one query.
 

If we can measure the first register of the final state in $\mathcal{A}'_m$, we can obtain an element $y \in K_m$. Once $y$ is nonzero, we can add {$y$} into $Y$ to get a larger known linearly independent set in $S^{\perp}$. Thus, we say Algorithm $\mathcal{A}'_m$ is successful if the measurement result of the first register of the final state is a nonzero element in $K_m$. The success probability of $\mathcal{A}'_m$ is $1-1/|K_m| = 1-1/p^{n-k-m}$. 

Furthermore, we wish to avoid obtaining zero element after measurement. Since $\frac{1}{\sqrt{1-1/p^{n-k-m}}} \le \frac{1}{\sqrt{1-1/p}} \le \sqrt{2}$, we use $O(\frac{1}{\sqrt{1-1/p^{n-k-m}}}) = O(1)$ applications of $\mathcal{A}'_m$ and ${\mathcal{A}'_m}^{-1}$ to get an algorithm $\mathcal{A}''_m$ that can obtain a nonzero element in $K_m$ with certain by the amplitude amplification technique of \cite{Brassard2002quantum}. ${\mathcal{A}'_m}^{-1}$ needs to use $O^{-1}_f$ and other operations not related to $f$. Since $O_f^2=I$ by definition, we have $O_f^{-1} = O_f$. 
Thus, $\mathcal{A}''_m$ makes $O(1)$ queries.

Finally, we gave Algorithm \ref{qumain} to solve $\mathsf{GSP}(p,n,k)$ exactly. During the iteration of each time, an element is added into $Y$ to consist of a larger linearly independent set in $S^{\perp}$. Finally, $\langle Y \rangle = S^{\perp}$.
\begin{algorithm}
\caption{An exact quantum algorithm for $\mathsf{GSP}(p,n,k)$}
		\label{qumain}
		\begin{algorithmic}[1]
			\STATE $Y = \emptyset$; 
			\FOR{$i =0\ to\ n-k-1$}  
			\STATE {Find $K_i$ using Algorithm \ref{sub2};}
			\STATE For a subgroup $K_i \le S^{\perp}$ satisfying that $S^{\perp} = K_i + \langle Y \rangle$,
			perform Algorithm $\mathcal{A}''_i$ to obtain a nonzero element $y_i \in K_i$;
			\STATE $Y = Y \cup \{y_i\}$;
			\ENDFOR
			\RETURN $\langle Y \rangle ^{\perp}$;
		\end{algorithmic}
	\end{algorithm}
Since $\mathcal{A}''_i$ makes $O(1)$ queries for $1 \le i \le k$, the total number of queries is $O(n-k)$.
\end{appendix}
\end{document}